\documentclass[12pt]{article}

\usepackage{amsmath}
\usepackage{amssymb}
\usepackage{amsfonts}
\usepackage{latexsym}
\usepackage{braket}
\usepackage{color}

\catcode `\@=11 \@addtoreset{equation}{section}
\def\theequation{\arabic{section}.\arabic{equation}}
\catcode `\@=12

\newcommand{\be}{\begin{equation}}
	\newcommand{\en}{\end{equation}}
\newcommand{\bea}{\begin{eqnarray}}
	\newcommand{\ena}{\end{eqnarray}}
\newcommand{\beano}{\begin{eqnarray*}}
	\newcommand{\enano}{\end{eqnarray*}}
\newcommand{\bee}{\begin{enumerate}}
	\newcommand{\ene}{\end{enumerate}}

\newcommand{\mc}{\mathcal}

\newcommand{\D}{{\mc D}}

\newcommand{\Sc}{{\cal S}}

\newcommand{\F}{{\cal F}}
\newcommand{\G}{{\cal G}}

\newcommand{\Lc}{{\cal L}}

\newcommand{\1}{1 \!\! 1}

\newcommand{\ltwo}{{\Lc^2(\mathbb{R})}}
\newcommand{\scr}{{\Sc(\mathbb{R})}}

\newcommand{\Hil}{\mc H}

\newtheorem{thm}{Theorem}

\newtheorem{lemma}[thm]{Lemma}
\newtheorem{prop}[thm]{Proposition}
\newtheorem{defn}[thm]{Definition}

\newenvironment{proof}{\noindent {\bf Proof --}}{\hfill$\square$ \vspace{3mm}\endtrivlist}

\renewcommand{\l}{\langle}
\renewcommand{\r}{\rangle}
\newcommand{\pin}[2]{\l#1 , #2\r}

\catcode `\@=11 \@addtoreset{equation}{section}
\catcode `\@=12

\begin{document}
	
	\thispagestyle{empty}

	\vspace*{2cm}
	
	\begin{center}
		{\Large \bf An unified approach for the (inverted) quantum harmonic oscillator and the Berry-Keating hamiltonian}   
		
		\vspace{5mm}
		
		{\large F. Bagarello}\\
		 Dipartimento di Ingegneria,
		Universit\`a di Palermo,\\ I-90128  Palermo, Italy\\
		and I.N.F.N., Sezione di Napoli\\
		e-mail: fabio.bagarello@unipa.it\\

		\vspace{8mm}
		
		{\large E. Balistreri}\\
		Dipartimento di Fisica e Chimica,
		Universit\`a di Palermo,\\ I-90128  Palermo, Italy\\
		e-mail: emanuele.balistreri@community.unipa.it\\

		\vspace{8mm}
		
		{\large S. Ku\.{z}el}\\
		AGH University, Krak\'ow, Poland\\
		e-mail: kuzhel@agh.edu.pl\\

	\end{center}
	
	\vspace*{1cm}
	
	\begin{abstract}
	\noindent We propose an unified view to three well known quadratic Hamiltonians widely considered in quantum mechanics, i.e. those of the harmonic oscillator, of the inverted harmonic oscillator, and the Berry-Keating Hamiltonian. We describe an unified approach to determine the discrete spectrum of these operators, and to deduce their (generalized) eigenstates. We also look at these Hamiltonians as the generators of $\mathfrak{su}(1,1)$, and of a generalized version of it relevant in presence of non Hermitian generators.
	\end{abstract}
	
	\vspace{2cm}
	
	%{\bf PACS Numbers}:  .......
	
	\vfill

	%\pagenumbering{roman}
	
	\newpage
	
	\section{Introduction}
	
	There is no need to discuss here the relevance of the harmonic oscillator (HO) in classical and quantum mechanics. We only want to mention that it is one of the (relatively) few systems whose equations of motion can be solved exactly and that it often appears as a first-order effect when setting up a perturbative approach to deal with a more complicated system. We refer to \cite{mer}-\cite{cohen} for many well known results on HO, and for the analysis of the eigenstates and eigenvalues of its Hamiltonian $H_{HO}$. 
	
	On the other hand, the inverted harmonic oscillator (IHO) is not as ubiquitous as the HO is. The reason is, we believe, that from a physical point of view the IHO describes an unstable system and, as such, it is often considered to be unphysical. However, despite this apparently negative view to this system, the Hamiltonian of the IHO, $H_{IHO}$, has been investigated in recent years by many authors also because of the interesting mathematical results which are naturally connected with it. We refer to \cite{barton}-\cite{bagiqho} as a very partial list of recent (and not so recent) articles on this topic. In particular in \cite{bagiqho} one of us has proposed a distributional point of view to the eigenvectors of $H_{IHO}$, and to its coherent states. 
	
	As pointed out by many authors, the Hamiltonian of the IHO  is also connected to another interesting Hamiltonian, $H_{BK}$,  considered by Berry and Keating in \cite{bk1,bk2} in connection with the zeros of the Riemann's function. This relation was briefly considered, for instance, in \cite{sunda}. Since \cite{bk1,bk2}, many authors have started to investigate the system they proposed, due to the fact that, despite its apparent simplicity, it produces a lot of interesting results. We refer to \cite{bbm, kalauni} for some papers on this aspect.
	
	Our paper aims to give an unifying view to these three systems which, in fact, are all described by quadratic Hamiltonians in $\hat q$ and $\hat p$, the self-adjoint position and momentum operators. In particular, we will show that this unique settings can be obtained in two different ways. In Section \ref{sect2} we will introduce a single Hamiltonian depending on some parameters, $\alpha_0$, $\alpha_1$ and $\alpha_2$, that, fixing specific choices of them, returns either $H_{BK}$, $H_{HO}$ or $H_{IHO}$. We will see that these different possibilities are connected to different mathematical settings, in particular in agreement with what is discussed in \cite{bagiqho}. In addition, in Section \ref{sect3} we will discuss the role of an extended scalar product recently proposed in \cite{bag2023} in connection with $H_{IHO}$ where, as we will see, the ordinary scalar product in $\ltwo$ cannot work. Interestingly enough, in Section \ref{sect4} we show that an Hilbert space can still be introduced even in the original distributional settings, and we explore this possibility.
	
	In Section \ref{sect5} we explore another, apparently different, approach to the same Hamiltonians. In particular, we show that they are proportional to the generators of the $\mathfrak{su}(1,1)$ algebra, or to a slightly extended version of it first proposed in \cite{bag2021JPA} and useful when the generators are not necessarily self-adjoint. The results we obtain in this way reflect those deduced in Sections \ref{sect2} and \ref{sect3}, but with a completely different approach.
	
	Section \ref{sect6} contains our conclusions and plans for the future. To make the paper self-contained, we also include two appendices, where some aspects of the paper are further discussed. 
	
	\section{Our objective and main result}\label{sect2}
	We start by introducing the following well-known operators:
\be
	H_0={\hat p}^2+{\hat q}^2=H_{HO},  \quad
	H_1={\hat p}{\hat q}+{\hat q}{\hat p}=H_{BK},
	\quad
	H_2={\hat p}^2-{\hat q}^2=H_{IHO}.
	\label{hamtot}\en
	We observe that there is an unessential $\frac{1}{2}$ missing above. We prefer this notation to simplify some of the formulas below. Of course, it is not a big problem to go back to the standard definition of operators. In addition, to simplify further the notation, we prefer to use $H_0$, $H_1$ and $H_2$ rather than $H_{HO}$, $H_{BK}$, $H_{IHO}$.
	Our main purpose in this paper is to analyze their linear combination
	\begin{equation}\label{new5}
		H_\alpha=\alpha_0H_0+\alpha_1H_1+\alpha_2H_2=(\alpha_0+\alpha_2)\hat{p}^2+\alpha_1(\hat{p}\hat{q}+\hat{q}\hat{p})+(\alpha_0-\alpha_2)\hat{q}^2, \quad  \alpha_i\in\mathbb{R},
	\end{equation}
	when the coefficients $\alpha_j$ change. In particular,  when $\alpha_1$ and $\alpha_2$ are {\em small} compared to $\alpha_0$, the operator $H_\alpha$ can be thought of as a perturbed version of the (usual) quantum harmonic oscillator $H_0$. Similarly if $\alpha_0$ and $\alpha_1$ are much smaller than $\alpha_2$, then $H_\alpha$ is essentially a perturbed version of $H_{IHO}$. And so on.
	
	The Hamiltonian $H_\alpha$ can be rewritten in a simpler way with a canonical transformation. In doing so, we will clearly describe the different regimes of the Hamiltonian, showing that different ranges of the parameters produce different situations. With respect to the original Hamiltonian in (\ref{new5}) the Hamiltonian we will deduce here has the advantage to be diagonal, even if (possibly) not bounded from below. 
	
	A straightforward computation shows that, defining
	\be 
	\hat Q=\hat q\,\cos\theta+\hat p\,\sin\theta, \qquad \hat P=-\hat q\,\sin\theta+\hat p\,\cos\theta, 
	\label{41}\en
	we have $[\hat Q,\hat P]=i I$ for all $\theta\in [0,2\pi[$. Now, if we restrict ourselves to $\theta\neq\frac{\pi}{4}, \frac{3\pi}{4}$, and in particular we take
	\be 
	\tan(2\theta)=-\frac{\alpha_1}{\alpha_2}, 
	\label{42}\en
	then, introducing two new parameters
	\be 
	\beta_0=\alpha_0+\frac{\alpha_2}{\cos2\theta}, \qquad \beta_1=\alpha_0-\frac{\alpha_2}{\cos2\theta},
	\label{43}\en
	we can rewrite $H_\alpha$ as follows:
	\be 
	H_\alpha=\beta_0\hat P^2+\beta_1\hat Q^2.
	\label{44}\en
	The computations are simple and will not be given here. We observe that our constraint of $\theta$ implies that $\alpha_2\neq0$, which is not a major problem, since it only means that the coefficients of $\hat q^2$ and $\hat p^2$ in (\ref{new5}) must be different. Equations (\ref{42}) and (\ref{43}) can be easily inverted, and we find
	\be 
	\alpha_0=\frac{1}{2}(\beta_0+\beta_1), \qquad \alpha_1=\frac{1}{2}(\beta_1-\beta_0)\sin(2\theta), \qquad \alpha_2=\frac{1}{2}(\beta_0-\beta_1)\cos(2\theta).
	\label{45}\en
	Let us now define $\xi=\alpha_0^2-\alpha_1^2-\alpha_2^2$. Here, we do not require $\xi$ to be positive. We have the following 
	
	\begin{thm}\label{thm1}
		If $\xi>0$ the Hamiltonian $H_\alpha$	is (plus or minus) a quantum harmonic oscillator in $\hat Q$ and $\hat P$. If $\xi=0$ then $H_\alpha$ is (essentially) a free particle, in the same new variables. Finally, if $\xi<0$, then $H_\alpha$ is an inverted quantum harmonic oscillator, again in $\hat Q$ and $\hat P$.
	\end{thm}
	\begin{proof}
		The proof immediately follows from the identity $\xi=\beta_0\beta_1$, and from the expression (\ref{44}) of $H_\alpha$.
	\end{proof}
	
	\vspace{2mm}
	{\bf Remarks:--} (1) If $\xi=0$ then either $\beta_0$ or $\beta_1$ (or both) are zero. If we consider the case of $\beta_1=0$, then $H_\alpha$ is a free particle in the new variable $\hat Q$. Analogously, if $\beta_0=0$, then $H_\alpha$ is again a free particle, but in the Fourier transform space. Of course, when both $\beta_0$ and $\beta_1$ are zero the situation is not really interesting: $H_\alpha=0$.
	
	(2) If $\xi<0$ $\beta_0$ and $\beta_1$ must have different signs. For instance, we can imagine that $\beta_0>0$ while $\beta_1<0$. From (\ref{44}) we see that $H_\alpha$ is, in fact, an inverted harmonic oscillator in $\hat Q$ and $\hat P$ which can be treated as in \cite{bagiqho}. 
	
	(3) Of course, we cannot use the representation (\ref{44}) of $H_\alpha$ to analyze the Berry-Keating Hamiltonian, since this can be recovered only if $\alpha_0=\alpha_2=0$, which creates some difficulties with formula (\ref{42}). However, the conclusions are similar to the previous ones, as we will show in what follows.
	
	\subsection{The case $\alpha_2=0$}
	
	In this case, $H_\alpha$ in (\ref{new5}) reads
	\begin{equation}\label{46}
		H_\alpha=\alpha_0\hat{p}^2+\alpha_1(\hat{p}\hat{q}+\hat{q}\hat{p})+\alpha_0\hat{q}^2, \quad  \alpha_i\in\mathbb{R},
	\end{equation}
	so that the coefficients of $\hat p^2$ and $\hat q^2$ coincide. In this case (\ref{41}) can still be used, but with $\theta=\frac{\pi}{4}$ if we want to avoid the case of $\beta_0=\beta_1$ necessarily. In this case $H_\alpha$ above can still be written as in (\ref{44}), but with the following relations between the $\alpha$'s and the $\beta$'s:
	\be \label{46b}
	\alpha_0=\frac{1}{2}(\beta_0+\beta_1), \qquad \alpha_1=\frac{1}{2}(\beta_1-\beta_0), 
	\en
	with, of course, $\alpha_2=0$, or, inverting these,
	\be 
	\beta_0=\alpha_0-\alpha_1, \qquad \beta_1=\alpha_0+\alpha_1.
	\label{47}\en
	Again we find that
	$$
	\xi=\beta_0\beta_1,
	$$
	 so that the conclusion is the same as above: the sign of $\xi$ is linked to the product between $\beta_0$ and $\beta_1$. For this reason, no real differences emerge with respect to what stated in Theorem \ref{thm1}.
	
	\vspace{2mm}
	
	{\bf Remark:--} The solution in (\ref{47}) could also be found by looking directly at (\ref{42}), and sending $\alpha_2$ to zero while keeping, in particular, $\frac{\alpha_1}{\alpha_2}$ negative. In this case (\ref{42}) admits $\theta=\frac{\pi}{4}$ as a possible solution, which is exactly what we have deduced before. Other solutions, of course, also exist.
	
	\vspace{2mm}
	
	It might be interesting to briefly discuss what happens when $H_\alpha$ can be seen as a perturbed version of $H_0$. For that we assume that $\alpha_1$ and $\alpha_2$ can be written as $\alpha_j=\alpha_0\epsilon_j$, $j=1,2$, so that $H_\alpha=\alpha_0(H_0+\epsilon_1H_1+\epsilon_2H_2)$ and, if $\epsilon_1$ and $\epsilon_2$ are small, then $H_\alpha/\alpha_0$ does not differ much from $H_0$. Here we will only consider non zero $\epsilon_j$. Then, from (\ref{42}) we get that $\tan(2\theta)=-\frac{\epsilon_1}{\epsilon_2}$, so that we can write $\sin(2\theta)=-\gamma\epsilon_1$ and $\cos(2\theta)=\gamma\epsilon_2$ for  $\gamma$ such that $\gamma^2=\frac{1}{\epsilon_1^2+\epsilon_2^2}$. Then we can rewrite
	\be 
	H_\alpha=\frac{\alpha_0}{\gamma}\left((\gamma+1)\hat P^2+(\gamma-1)\hat Q^2\right).
	\label{47b}\en
	This shows the following: $\gamma>1$ corresponds to a harmonic oscillator, while $\gamma=1$ returns a free particle and $\gamma<1$ an inverted oscillator. This is in agreement with our view of $H_\alpha$ as a perturbed version of $H_0$. In fact, if $\gamma\leq1$, then $\epsilon_1^2+\epsilon_2^2\geq1$, which means that at least one of the $\epsilon_j$ should be {\em large enough}. Hence either $H_1$ or $H_2$, or both, are "comparable" with $H_0$, and therefore they cannot be thought to be perturbations of $H_0$. On the other hand, this is possible when $\gamma>1$ since, in this case, $\epsilon_1^2+\epsilon_2^2<1$ implies that 
	$\epsilon_1$ and $\epsilon_2$ cannot be large.

	\section{The analysis of $\beta_0\hat P^2+\beta_1\hat Q^2$}\label{sect3} 
	
	This section is devoted to the analysis of the operator in (\ref{44}) which, for obvious reasons, from now on we will indicate with $H_\beta$ rather than $H_\alpha$:
	\be 
	H_\beta=\beta_0\hat P^2+\beta_1\hat Q^2,
	\label{51}\en
	considering the case of $\beta_0\beta_1$ positive or negative. For concreteness, and since this is not a major constraint, we will focus on the following situations: 
	
	{\bf case nr. 1} (HO): $\beta_0>0$ and $\beta_1>0$, and 
	
	{\bf case nr. 2} (IHO): $\beta_0>0$ and $\beta_1<0$.
	
	Here HO and IHO stand for, respectively, {\em Harmonic Oscillator} and {\em Inverted Harmonic Oscillator}, as for the Hamiltonians in (\ref{hamtot}). In both cases we need to introduce the positive quantity $\Omega^2=\frac{|\beta_1|}{\beta_0}$.
	
	\vspace{2mm}
	
	The case nr. 1 is very well known: if we introduce the bosonic operators
	\be 
	c=\sqrt{\frac{\Omega}{2}}\left(\hat Q+i\frac{\hat P}{\Omega}\right), \qquad c^{\dag}=\sqrt{\frac{\Omega}{2}}\left(\hat Q-i\frac{\hat P}{\Omega}\right),
	\label{52}\en

	we can rewrite $H_\beta$ in the following form:
	\be 
	H_\beta=2\beta_0\Omega\left(c^\dagger c+\frac{1}{2}\1\right)=2\beta_0\tilde H_{HO},
	\label{53}\en
	with obvious notation.
	Here and in the following, we will always assume that $\Omega>0$. Of course, $[c,c^\dagger]=\1$. The eigenstates and the eigenvalues of $H_\beta$ are the standard ones:
	\be 
	\varphi_n(x)=\frac{1}{\sqrt{2^n\,n!}}\left(\frac{\Omega}{\pi}\right)^{1/4}H_n\left(\sqrt{\Omega}\,x\right)e^{-\frac{1}{2}\Omega x^2},
	\label{54}\en
	where $H_n(x)$ is the $n$-th Hermite polynomial\footnote{It is clear that $x$ is the variable associated to $\hat Q$, meaning with this that $\hat Q g(x)=xg(x)$ for all $g\in D(\hat Q)$, the domain of the operator $\hat Q$.}, with
	\be 
	H_\beta \varphi_n(x)=E_n\varphi_n(x), \qquad E_n=2\beta_0\Omega\left(n+\frac{1}{2}\right)=2\sqrt{\beta_0\beta_1}\left(n+\frac{1}{2}\right).
	\label{55}\en
	It is clear that $\varphi_n(x)\in\ltwo$ for all $n=0,1,2,3,\ldots$. It is also well known that $\langle\varphi_n,\varphi_m\rangle=\delta_{n,m}$.

	\vspace{3mm}
	
	The case nr. 2 is probably less known, and less discussed in the literature, even if it has attracted some interest in recent years. In particular, in \cite{bagiqho} it was discussed in connection with the so-called {\em weak pseudo-bosons}, \cite{bagspringerbook}.
	
	In this case, we have
	\be 
	H_\beta=2\beta_0\tilde H_{IHO}, \qquad\mbox{where }\quad  \tilde H_{IHO}=\frac{1}{2}\left(\hat P^2-\Omega^2\hat Q^2\right). 
	\label{56}\en
	If we introduce the operators
	\be
	A_\pm=\frac{1}{\sqrt{2\Omega}}\left(e^{\pm i\pi/4}\Omega\,\hat Q+i\,e^{\mp i\pi/4}\hat P\right), \,\,\, B_\pm=\frac{1}{\sqrt{2\Omega}}\left(e^{\pm i\pi/4}\Omega\,\hat Q-i\,e^{\mp i\pi/4}\hat P\right),
	\label{57}\en
	and the functions 
	\be
	\varphi_n^{(\pm)}(x)=\frac{e^{\pm i\pi/8}}{\sqrt{2^n\,n!}}\left(\frac{\Omega}{\pi}\right)^{1/4}H_n\left(e^{\pm i\pi/4}\sqrt{\Omega}\,x\right)e^{\mp\,\frac{i}{2}\,\Omega \,x^2}
	\label{58}\en
	with $\psi_n^{(\pm)}(x)=\varphi_n^{(\mp)}(x)$, the following facts and results are proved in \cite{bagiqho}:
	
	\begin{enumerate}
		\item 
		\be
		B_\pm^\dagger=A_\mp, \qquad A_\pm^\dagger=B_\mp,
		\label{59}\en
		and
		\be
		A_\pm=\frac{c\pm i c^\dagger}{\sqrt{2}}, \qquad B_\pm=\frac{c^\dagger\pm i c}{\sqrt{2}}.
		\label{510}\en
		All of these operators leave $\scr$ stable. In particular,
		\be
		[A_\pm,B_\pm]f=f,
		\label{511}\en
		for all $f\in\scr$. Hence, the pairs $(A_\pm,B_\pm)$ are pseudo-bosonic, \cite{bagspringerbook}.

		\item The following ladder equations are satisfied:
		\be
		\left\{
		\begin{array}{ll}
			A_\pm\,\varphi_0^{(\pm)}=0,\qquad A_\pm\,\varphi_n^{(\pm)}(x)=\sqrt{n}\,\varphi_{n-1}^{(\pm)}, \hspace{2cm} n\geq 1,\\
			B_\pm\,\varphi_n^{(\pm)}=\sqrt{n+1}\,\varphi_{n+1}^{(\pm)}, \hspace{5cm} n\geq 0,
		\end{array}
		\right.
		\label{512}\en
		and
		\be
		\left\{
		\begin{array}{ll}
			B_\pm^\dagger\,\psi_0^{(\pm)}=0,\qquad B_\pm^\dagger\,\psi_n^{(\pm)}(x)=\sqrt{n}\,\psi_{n-1}^{(\pm)}, \hspace{2cm} n\geq 1,\\
			A_\pm^\dagger\,\psi_n^{(\pm)}=\sqrt{n+1}\,\psi_{n+1}^{(\pm)}, \hspace{5cm} n\geq 0.
		\end{array}
		\right.
		\label{513}\en
		
		\item The Hamiltonian $\tilde H_{IHO}$ can be rewritten in terms of the operators $A_\pm$ and $B_\pm$ as follows:
		\be 
		\tilde H_{IHO}=i\Omega\left(B_+ A_++\frac{1}{2}\,\1\right)=-i\Omega\left(B_- A_-+\frac{1}{2}\,\1\right).
		\label{514}\en
		
		\item 
		\be
		\tilde H_{IHO} \varphi_n^{(\pm)}=\pm i\Omega \left(n+\frac{1}{2}\right)\varphi_n^{(\pm)}, \quad n\geq0.
		\label{515}\en

		 From (\ref{56}) we deduce that
		\be
		H_{\beta} \varphi_n^{(\pm)}=\pm 2i\beta_0\Omega \left(n+\frac{1}{2}\right)\varphi_n^{(\pm)}, \quad n\geq0.
		\label{516}\en
	  Then, in case nr. 2 the originally real and positive eigenvalues of $H_\beta$ turn into purely imaginary numbers, uniformly distributed along the imaginary axis.

	\end{enumerate}

	\vspace{2mm}
	
	{\bf Remark:--} It is easy to construct a formal map from $\tilde H_{HO}$ to $\tilde H_{IHO}$. It is sufficient to consider an operator (actually, two) $T_\pm$ acting on $\Omega$ as follows: 
	\be
	T_\pm(\Omega)=e^{\pm i\frac{\pi}{2}}\Omega=\pm i\Omega.
\label{T1}\en
	Then we have $T_\pm(\tilde H_{HO})T^{-1}_\pm=\tilde H_{IHO}$ and $T_\pm[\varphi_n(x)]=\varphi_n^{(\pm)}(x)$. In line with these results, we observe that
	$$
	\tilde H_{IHO}\varphi_n^{(\pm)}=T_\pm \left(\tilde H_{HO}\varphi_n\right)=T_\pm \left(E_n\varphi_n\right)=\pm i E_n \varphi_n^{(\pm)},
	$$
	as we expect. In fact, this is exactly formula (\ref{515}). We will return to $T_\pm$,  and to their role in our analysis in Section \ref{sect3}, adopting a rigorous approach to this operator, which indeed has a very interesting interpretation and role, as we will see.
	
	\vspace{2mm}
	
	In \cite{bagiqho} the nature of the extended eigenfunctions $\varphi_n^{(\pm)}(x)$ has been investigated in the context of distribution theory. In particular, the following quantities have been introduced for $n\geq0$:
	\be
	\Phi_n^{(\pm)}[f]=\langle\varphi_n^{(\pm)},f\rangle, \quad \Psi_n^{(\pm)}[g]=\langle\psi_n^{(\pm)},g\rangle, \qquad f, g\in\scr.
	\label{517}\en
	  Here $\langle.,.\rangle$ is the extension of the ordinary scalar product to {\em compatible pairs}, i.e. to pairs of functions with are, when multiplied together, integrable, but separately they are not (or, at least, one is not). Then, \cite{bagiqho}, one can prove  that $\Phi_n^{(\pm)}[f]$ and  $\Psi_n^{(\pm)}[g]$ are well defined, linear, and continuous in the natural topology $\tau_\Sc$ in $\scr$. In a few words, they are tempered distributions: $\Phi_n^{(\pm)}, \Psi_n^{(\pm)}\in\Sc'(\mathbb{R})$.
	
	This result suggests that it does not make much sense to ask if different eigenvectors are mutually orthogonal. In fact, it is clear that, for instance, $\langle\varphi_n^{(+)},\varphi_m^{(+)}\rangle$ is not convergent. Indeed, this is essentially the integral in $\mathbb{R}$ of a polynomial of degree $n+m$. However, in some recent papers, \cite{bag2023,bag2025}, a possible extended scalar product has been proposed, which coincides with the ordinary product in $\ltwo$ for square-integrable functions, but can also be extended to other functions or even to some distributions.
	
	The idea is the following: let us consider an orthonormal (o.n) basis in $\ltwo$,
	$$
    \F_e=\{e_{n}\in \Sc(\mathbb{R}), \ \, \ n\geq0\}.
    $$ In what follows, for concreteness, $\F_e$ will be taken as the set of eigenstates of the quantum harmonic oscillator\footnote{Hence $e_n(x)$ essentially coincides with the functions $\varphi_n(x)$ in (\ref{54}). We use a different notation to stress the fact that $e_n(x)$ here is {\bf not} necessarily the o.n. basis of an harmonic oscillator.} with ladder operators $c$ and $c^\dagger$ in (\ref{52}). Using the Parseval identity, we have that for all $f, g\in \ltwo$, 
	$$
	\pin{f}{g}=\sum_{n}\pin{f}{e_n}\pin{e_n}{g}=\sum_{n} \overline{f}[e_n]\,g[\overline{e_n}]=\sum_{n} \overline{f}[e_n]\,g[{e_n}],
	$$
	since each $e_n$ is real. Here we use the notation
	\be
	\overline{h}[c]=\int_{\mathbb{R}} \overline{h(x)}\,c(x)\, dx=\pin{h}{c}.
	\label{518}\en
	Of course, in the Parseval identity above, the particular choice of $\F_e$ is not relevant, as far as $f, g\in \ltwo$.  Now, since $e_n\in\Sc(\mathbb{R})$, the following quantity is well defined:
	\be
	\overline{K}[e_n]=\pin{K}{e_n},
	\label{23}\en
	for all $K\in\Sc'(\mathbb{R})$.
	What might exist or not is the following sum
	\be
	\pin{F}{G}_e=\sum_{n}\overline{F}[e_n]G[e_n],
	\label{24}\en
	where $F, G\in\Sc'(\mathbb{R})$ are two generic tempered distributions. For this reason, in \cite{bag2023}  the following definition has been proposed:
	
	\begin{defn}\label{def1}
		Two tempered distributions $F, G\in\Sc'(\mathbb{R})$ are $\F_e$-multiplicable if the series in (\ref{24}) converges. In this case, we call $\pin{F}{G}_e$ the $e$-product of $F$ and $G$.
	\end{defn}
	
	We refer to \cite{bag2023,bag2025} for more results and examples of this multiplication. Here we want to show that (at least) some of the $\varphi_n^{(\pm)}$'s admit a finite $e$-multiplication. More explicitly, we will show that
	\be 
	\pin{\varphi_0^{(+)}}{\varphi_0^{(+)}}_e=\infty, \qquad \mbox{ but }\qquad \pin{\varphi_0^{(+)}}{\varphi_1^{(+)}}_e=0.
	\label{519}\en
	We will also briefly discuss how to extend these results to other eigenstates.
	
	The first useful fact is that, using (\ref{510}) and $A_+\,\varphi_0^{(+)}=0$ we find that $c\,\varphi_0^{(+)}=-i\,c^\dagger\,\varphi_0^{(+)}$. Then, with simple computations, we deduce the following:
	\be
	\pin{\varphi_0^{(+)}}{e_l}=\left\{
	\begin{array}{ll}
		0, \hspace{7cm} l=2k+1,\\
		i^k\sqrt{\frac{(2k-1)!!}{(2k)!!}}\pin{\varphi_0^{(+)}}{e_l}, \hspace{4cm} l=2k,
	\end{array}
	\right.
	\label{520}\en
	where $k=0,1,2,\ldots$. Then we have
	\be 
	\pin{\varphi_0^{(+)}}{\varphi_0^{(+)}}_e=\sum_{l=0}^\infty|\pin{\varphi_0^{(+)}}{e_l}|^2=\sum_{k=0}^\infty|\pin{\varphi_0^{(+)}}{e_{2k}}|^2=|\pin{\varphi_0^{(+)}}{e_{0}}|^2\left(1+\sum_{k=1}^\infty\frac{(2k-1)!!}{(2k)!!}\right)=\infty,
	\label{521}\en
	as we can check using the Gauss test for series. The conclusion is that in this case the ordinary scalar product in $\ltwo$ does not differ from the $e$-multiplication.
	
	To prove the second result in (\ref{519}) we first observe that 
	$$
	\pin{\varphi_0^{(+)}}{\varphi_1^{(+)}}_e=\sum_{l=0}^\infty\pin{\varphi_0^{(+)}}{e_n}\pin{e_n}{\varphi_1^{(+)}}=\sum_{k=0}^\infty\pin{\varphi_0^{(+)}}{e_{2k}}\pin{e_{2k}}{\varphi_1^{(+)}},
	$$
	using the fact that, see (\ref{520}), $\pin{\varphi_0^{(+)}}{e_{2k+1}}=0$ for all $k$. Now, since $\varphi_1^{(+)}=B_+\varphi_0^{(+)}=\frac{1}{\sqrt{2}}(c^\dagger+ic)\varphi_0^{(+)}$, we find that
	$$
	\pin{e_{l}}{\varphi_1^{(+)}}=\left\{
	\begin{array}{ll}
		\frac{i}{\sqrt{2}}\pin{e_1}{\varphi_0^{(+)}}=0, \hspace{7cm} l=0,\\
		\frac{1}{\sqrt{2}}\left(\sqrt{l}\,\pin{e_{l-1}}{\varphi_0^{(+)}}+i\sqrt{l+1}\pin{e_{l+1}}{\varphi_0^{(+)}}\right), \hspace{2.7cm} l\geq1,
	\end{array}
	\right.
	$$
	using again  (\ref{520}). It is now clear that, if $l=2k$, then $l\pm1$ are odd, and using once more (\ref{520}), we conclude that $\pin{e_{2k}}{\varphi_1^{(+)}}=0$ for all $k=0,1,2,3,\ldots$. Hence, our claim follows. 
	
	The same trick can be adopted to compute, for instance 	$\pin{\varphi_0^{(-)}}{\varphi_1^{(-)}}_e$, or other extended scalar products, but we will not do it here. We just want to comment that the role of distributions will also be relevant in case of $\beta_0\beta_1=0$ since, in this case, the eigenstates of $H_\beta$ are noting than plane wave (in $x$ or in the Fourier transform). Then again, $\ltwo$ is not the appropriate space to work with.
	
	\vspace{2mm}
	
	{\bf Remark:--} The Berry-Keating Hamiltonian can be recovered setting $\alpha_0=0$ in (\ref{46}). This means, see (\ref{46b}), that $\beta_0=-\beta_1$, so that $\alpha_1=\beta_1$. Hence, the Hamiltonian {$	H_\alpha=\alpha_1(\hat{p}\hat{q}+\hat{q}\hat{p})$} can be rewritten as $H_\beta=\beta_0\left(\hat P^2-\hat Q^2\right)$: $\Omega=\Omega^2=1$: the Berry-Keating Hamiltonian is just the Hamiltonian for the inverted quantum oscillator, but in different variables.
	
	\section{The operators $T_\pm$}\label{sect4}
	
	In this section, we will return to the operators $T_\pm$ considered in (\ref{T1}), and we will describe, in a rigorous way, what we can derive from them. 
	
	In this perspective, we define $T_\pm$ by their action on the vectors $\varphi_n$ of the o.n. basis of $H_0$ in (\ref{hamtot}):
	\be 
	T_\pm[\varphi_n]=\varphi_n^{(\pm)}, \quad n\geq0.
	\label{61}\en
	 Since $\F_\varphi=\{\varphi_n\}$ is an o.n. basis for $\ltwo$, $T_\pm$ are both densely defined on $\Lc_\varphi=l.s.\{\varphi_n\}$, the linear span of the $\varphi_n$'s. However, based on what we have seen before, the range of $\Lc_\varphi$ is not in $\ltwo$. If we call $\Sigma_\pm=l.s.\{\varphi_n^{(\pm)}\}$, it follows that $T_\pm:\Lc_\varphi\rightarrow\Sigma_\pm$, which can be identified as subspaces of $\Sc'(\mathbb{R})$. It is natural to introduce a kind of inverse mapping of $T_\pm$. We call these maps $V_\pm$ and we assume that
	\be 
	V_\pm[\varphi_n^{(\pm)}]=\varphi_n, \quad n\geq0.
	\label{62}\en
	 Of course, it does not make sense to analyze the density of the domain of definition of these operators. However, it is clear that
	\be 
	T_\pm V_\pm=\1_{\Sigma_\pm}, \qquad V_\pm T_\pm=\1_{\Lc_\varphi},
	\label{63}\en
	where $\1_{\Sigma_\pm}$ and $\1_{\Lc_\varphi}$ are the identity operators on $\Sigma_\pm$ and $\Lc_\varphi$, respectively. 
	
	\vspace{2mm}

	{\bf Remark:--} Because of (\ref{T1}) one can argue, at least on a formal level that, $V_{\pm}=T_{\mp}$. However, to transform this into a rigorous statement some work needs to be done. However, since this aspect is not so relevant here, we will not really consider this aspect in detail.
    
%	{\color{red}
%    In my view equation (3.17) it's not fully rigourous. Treat the action of an operator as action on the only parameter \(\Omega\)
%     seems a bit hazardous to me. 
%    Morover, in my opinion, the result presented here it's not consistent with the domain of the operators.
%    \(T_{+}:L_{\varphi}\rightarrow\Sigma_{+}\)
%    the inverse should be defined on \(\Sigma_{+}\), but \(T_{-}\) it's definend on \(L_{\varphi}\), that has only trivial intersection with \(\Sigma_{+}\).
%    The result can be estabilished if we extend this two operators on al larger space, that contain \(L_{\varphi}\) and \(\Sigma_{\pm}\). This is the case of \(\mathcal{G}\) defined in the appendix. We can try to find two operators \(\tilde{T}_{\pm}=\mathcal{G}\rightarrow\mathcal{G}\), such that the restriction of \(\tilde{T}_{\pm}\) on \(L_{\varphi}\) it's equal to \(T_{\pm}\). At this point the only thing that we have to prove is that the restriction on \(\tilde{T}_{\pm}\) on \(\Sigma_{\mp}\), coincides with \(V_{\mp}\).
%    My suggestion for the operator \(\tilde{T}_{\pm}\) remain \(\tilde{T}_{\pm}={exp(\mp\frac{\pi}{8}H_{BK})}\)
%    
%    
%    }
	\vspace{2mm}
	
	The operators $V_\pm$ allow us to introduce two maps respectively  on $\Sigma_\pm$ as follows:
	\be 
	(\Phi_\pm,\Psi_\pm)_\pm=\langle V_\pm\Phi_\pm,V_\pm\Psi_\pm\rangle, \qquad \Phi_\pm, \Psi_\pm\in\Sigma_\pm.
	\label{64}\en
	 Here $\langle.,.\rangle$ is the usual scalar product in $\ltwo$. It is easy to see that $(.,.)_+$ is a scalar product on $\Sigma_+$. Analogously, $(.,.)_-$ is a scalar product on $\Sigma_-$. In particular, for instance, we have that $(\Phi_+,\Phi_+)_+=\langle V_+\Phi_+,V_+\Phi_+\rangle\geq0$,  for all $\Phi_+\in\Sigma_+$, clearly. Moreover, $(\Phi_+,\Phi_+)_+=0$ if and only if $V_+\Phi_+=0$. But, acting to the left of this equality with $T_+$ and using (\ref{63}), we conclude that $\Phi_+=0$. 
	
	This result has two interesting consequences. The first is that the various $\varphi_n^{(\pm)}$ are o.n. with respect to $(.,.)_\pm$:
	\be 
	(\varphi_n^{(\pm)},\varphi_m^{(\pm)})_\pm=\langle V_\pm\varphi_n^{(\pm)},V_\pm\varphi_m^{(\pm)}\rangle=\langle \varphi_n,\varphi_m\rangle=\delta_{n,m},
	\label{65}\en
	which, in a sense, generalizes what we have found using the $e$-multiplication. 
	
	The second consequence is that $(.,.)_\pm$ define two norms $\|.\|_\pm$ on, respectively, $\Sigma_\pm$, so that their completions
	\be 
	\Hil_\pm=\overline{\Sigma_\pm}^{\|.\|_\pm}
	\label{66}\en
are indeed Hilbert spaces, containing distributions among their elements. 

The conclusion is therefore that the eigenstates of $H_\beta$ in case nr. 2 can still be seen as elements of some Hilbert space, but endowed with a different scalar product. The fact that the eigenvalues of $H_\beta$ are purely imaginary follows from the nature of the Hamiltonian, which looks formally self-adjoint with respect to $\langle.,.\rangle$, but loses its (rigorous, or even formal) self-adjointness with respect to $(.,.)_\pm$.

\section{$H_j$ as generators of Lie algebras}\label{sect5}

In this section, we propose a different point of view to the operators $H_j$, $j=0,1,2$, showing that they are strictly connected to some interesting algebras, and in particular to $\mathfrak{su}(1,1)$, and to its extended version, firstly considered, in our knowledge, in \cite{bag2021JPA} in the context of non-Hermitiam quantum mechanics. Many of the results derived here and in previous papers on these operators can be restated and reformulated in this Lie-algebraic approach.

The starting point are the following commutation rules among $H_j$:
\begin{equation}
	[H_{0},H_{1}]=-4iH_{2}, \quad
	[H_{1},H_{2}]=4iH_{0} \quad 
	[H_2,H_{0}]=-4iH_{1} \quad 
	\label{la1}
\end{equation}
where indeed all these formulas should be understood acting, for instance, on functions in $\Sc(\mathbb{R})$, which is stable under the action of these different Hamiltonians. This is not the only choice of an invariant set of functions where $H_j$ are all defined. Another such space is described in Appendix \ref{appendixB}. Still another possible approach would be to construct, if possible, a *-algebra of unbounded operators as, e.g., in \cite{bag2021JPA}, having each $H_j$ as its elements. In this case, $H_jH_k$ makes sense, and it is an element of the same algebra. However, we will not follow this idea: we prefer to work on the set $\Sc(\mathbb{R})$. 

Formulas (\ref{la1}) suggest the introduction of the following operators:
\begin{equation}
	K_1=\frac{1}{4}H_{1}, \qquad 
	K_2=\frac{1}{4}H_{2}, \qquad 
	K_3=-\frac{1}{4}H_{0}.   
	\label{la2}
\end{equation}
Hence, they are the generator of a ${\bf su}(1,1)$ algebra, since they obey the following commutators (which, again, we must interpret properly, i.e., acting on test functions):
\be
[K_1,K_2]=-iK_3,\qquad [K_2,K_3]=iK_1, \qquad [K_3,K_1]=iK_2.
\label{la3}\en
Notice that, at least formally, $K_j=K_j^\dagger$, $j=1,2,3$. Following the standard procedure for  ${\bf su}(1,1)$, we first define
\be
K_0=K_3, \qquad K_\pm=K_1\pm i K_2, \qquad K^2=K_0^2-K_0-K_+K_-,
\label{la4}\en
 and we observe that
 \be
 K_0^\dagger=K_0, \quad K_\pm^\dagger =K_\mp, \quad [K_0,K_\pm]=\pm K_\pm, \quad [K_+,K_-]=-2K_0, \quad [K^2,K_j]=0,
 \label{la5}\en
	$j=1,2,3,\pm$. It is well, known, but it is useful to stress, that despite the notation, $K^2$ is not a positive operator, as its definition clearly suggests. The reason for calling this operator $K^2$ in the literature follows from a parallel with the algebra of the angular momentum, $\mathfrak{su}(2)$. Because, in particular, $[K^2,K_3]=0$, we can look to simultaneous eigenstates of these two operators. We refer to \cite{adams} for details on $\mathfrak{su}(1,1)$. Here we only recall that, calling $\Phi_{j,q}$ a joint eigenvector of $K^2$ and $K_0$, one usually writes
	\be
	K^2\Phi_{j,q}=j(j+1)\Phi_{j,q},\qquad K_0\Phi_{j,q}=q\Phi_{j,q},
	\label{la6}\en	
	where $j$ is not necessarily positive, due to the fact that $K^2$ is not a positive operator. $K_\pm$ behave as ladder operators in the following sense: if $K_\pm\Phi_{j,q}\neq0$, then
	\be
	K_+\Phi_{j,q}=(q-j)\Phi_{j,q+1}, \qquad K_-\Phi_{j,q}=(q+j)\Phi_{j,q-1}.
	\label{la7}\en
	The set of eigenvalues of $K_0$ could be bounded from above and below, bounded from above and unbounded below, unbounded from above and bounded below, or fully unbounded, \cite{adams}. In our case, with the choice in (\ref{la2}), we will show in a moment that we are in the second situation: there is a maximum eigenvalue for $K_0$, but no minimum. To show this, we start introducing the operators \(a=\frac{1}{\sqrt{2}}(\hat q+i\hat p)\) and \(a^{\dag}=\frac{1}{\sqrt{2}}(\hat q-i\hat p)\), which are, in suitable units, the standard creation and annihilation operators of the harmonic oscillator\footnote{These operators  are slightly different from those in (\ref{52}), of course. However, the two pairs $(a,a^\dagger)$ and $(c,c^\dagger)$ both satisfy the  CCR.}. We know that $[a,a^\dagger]=\1$, again on $\scr$. Moreover, simple computations show that
	\be
	K_+=-\frac{i}{2}\,a^2, \qquad K_-=\frac{i}{2}\,{a^\dagger}^2, \qquad K_0=-\frac{1}{4}(2\hat n_0+\1) \qquad K^2=-\frac{3}{16}\1,
	\label{la8}\en
	where $\hat n_0=a^\dagger a$, showing that $K_+^\dagger=K_-$, as expected, and that $K^2<0$. Hence the only possible values of $j$ are those satisfying $j(j+1)=-\frac{3}{16}$, i.e. $j_1=-\frac{1}{4}$ and $j_2=-\frac{3}{4}$. It is easy to find a set of eigenstates of $K_3$, since this is proportional to the Hamiltonian $H_0$. We start considering a vector $e_0$ such that $ae_0=0$, and then we construct (as usual), $e_n=\frac{{a^\dagger}^n}{\sqrt{n!}}\,e_0$, $n\geq 1$. This is nothing but the set $\F_e$ introduced in Section \ref{sect3}. Now, since $\hat n_0e_n=ne_n$, we see that, in particular, $K_0e_0=-\frac{1}{4}e_0$, which can be identified with the second equation in (\ref{la6}) by putting
	\be
	\Phi_{j,q}\quad\rightarrow\quad \Phi_{j_k,-\frac{1}{4}}=e_0, 
	\label{la9}\en
	where $k=1,2$. Hence $q=-\frac{1}{4}$ is the choice of $q$ corresponding to a vector, $\Phi_{j_k,-\frac{1}{4}}$, which is the eigenvector of $K_0$ with the highest eigenvalue. Indeed, using (\ref{la7}) and (\ref{la8}), we have that
	$$
	K_+\Phi_{j_k,-\frac{1}{4}}=0,
	$$
	since $a^2e_0=0$. Of course, this implies, see again (\ref{la7}), that {\em we cannot go up} from $q$ to $q+1$, $q+2$, etc... On the other hand, since $K_-$ is proportional to ${a^\dagger}^2$, it is clear that $K_-e_0$ is proportional to $e_2$ and, more in general, that $K_-^le_0$ is proportional to $e_{2l}$, $\forall l\geq0$. This means that
	\be
	K_0\left(K_-^le_0\right)=-\left(l+\frac{1}{4}\right)\left(K_-^le_0\right).
	\label{la10}\en
	We see that the eigenvalues decrease to $-\infty$ when $l$ increases. It remains to understand what happens with the other vectors $e_{2n+1}$. They are not part of those appearing in (\ref{la10}), in particular, so that we might wonder that maybe they are not eigenstates of $K_0$. This is indeed not the case, as we will show now.  Indeed, it is sufficient to observe that, while $e_0$ satisfies $K_0e_0=-\frac{1}{4}e_0$, the vector $e_1$ is also an eigenstate of $K_0$, but with eigenvalue $-\frac{3}{4}$: $K_0e_1=-\frac{3}{4}e_1$. Hence, as in (\ref{la9}), we put
	\be
	\Phi_{j_k,-\frac{3}{4}}=e_1, 
	\label{la11}\en
	where $k=1,2$. Again, it is clear that $K_+\Phi_{j_k,-\frac{3}{4}}=0$, and that $K_-^le_1$ is proportional to $e_{2l+1}$, $\forall l\geq0$.  Hence
	\be
	K_0\left(K_-^le_1\right)=-\left(l+\frac{3}{4}\right)\left(K_-^le_1\right),
	\label{la12}\en
	and again the eigenvalues decrease to $-\infty$ when $l$ increases. In this way, the set
	$$
	\{K_-^le_0\}\cup \{K_-^le_1\}=\{K_-^l\Phi_{j_k,-\frac{1}{4}}\}\cup \{K_-^l\Phi_{j_k,-\frac{3}{4}}\}
	$$
	returns essentially (except for inessential normalization factors) the o.n. basis $\F_e$ of functions in $\scr$. In other words, all the functions $K_-^le_0$ and $K_-^le_1$, $l\geq0$ are in $\ltwo$. This is because\footnote{This aspect will be more clear later in this section, when considering a different choice of the $K_j$.} the operator $K_0$ is nothing but the Hamiltonian $H_0$ of the harmonic oscillator (expect for an overall constant), whose eigenstates are well known to be square-integrable functions.

\subsection{A different choice of the generators}\label{sectLAb}

In (\ref{la2}) we have considered $K_3=K_0$ to be essentially the Hamiltonian of the harmonic oscillator $H_0$. Indeed, we have defined $	K_3=-\frac{1}{4}H_{0}$. Therefore, it is not a big surprise that the spectrum of $K_3$ is negative and bounded above, but not below. If we are more interested in the analysis of, say, the inverted harmonic oscillator, then we should define $K_3$ in a different way, by putting it proportional to $H_2$, and try to repeat what we have done before. However, there is a simple no-go result that prevents us from doing this {\em in a simple way}. To be more precise, if we replace (\ref{la2}) assuming that $K_3$ is proportional to $H_2$, $K_1$ is proportional to $H_0$ or to $H_1$, while $K_1$ is proportional either to $H_1$ or to $H_0$, there is no solution satisfying (\ref{la3}) and being (even only formally) Hermitian. If we want to relate $H_j$ with operators $K_j$ satisfying (\ref{la3}), the only choice is to abandon the hermiticity of $K_j$. Indeed, if we introduce the operators 
	\begin{equation}
		L_1=\frac{i}{4}H_{0} \qquad 
		L_2=\frac{1}{4}H_{1} \qquad 
		L_3=\frac{i}{4}H_{2},   
		\label{la13}
	\end{equation}
	they {\em look} as a generator of a $\mathfrak{su}(1,1)$ algebra, since they obey the following commutators (which, as usual, we must interpret properly):
	\be
	[L_1,L_2]=-iL_3,\qquad [L_2,L_3]=iL_1, \qquad [L_3,L_1]=iL_2.
	\label{la14}\en
	However,  $L_j\neq L_j^\dagger$, $j=1,3$, while $L_2=L_2^\dagger$. Hence we are not really dealing with  ${\bf su}(1,1)$, but with a slightly different version of it which was first considered, to our knowledge, in \cite{bag2021JPA}. We list a few useful results of this algebra in the Appendix \ref{appendixC}.  
	
	As we will see here, the approach we are going to discuss is not particularly different from the one we have proposed before, but the conclusions are different, and reflect what was done before in this paper and in \cite{bagiqho}.

	We start by defining
	\be
	L_0=L_3, \qquad L_\pm=L_1\pm i L_2, \qquad L^2=L_0^2-L_0-L_+L_-,
	\label{la4}\en
	and we observe that
	\be
	L_0^\dagger=-L_0, \quad L_\pm^\dagger =-L_\pm, \quad [L_0,L_\pm]=\pm L_\pm, \quad [L_+,L_-]=-2L_0, \quad [L^2,L_j]=0,
	\label{la15}\en
	$j=1,2,3,\pm$. We see that there are similarities but also differences with what we have seen in (\ref{la5}). For example, $L_+$ is not an adjoint of $L_-$. And $L_0$ is not self-adjoint. On the other hand, the commutation rules between $L_0$ and $L_\pm$ reflect those between $K_0$ and $K_\pm$, and this allows us to repeat the same construction as before. In particular, since $[L^2,L_0]=0$, we can look for simultaneous eigenstates of these two operators. Again, we write
	\be
	L^2\Psi_{j,q}=j(j+1)\Psi_{j,q},\qquad L_0\Psi_{j,q}=q\Psi_{j,q},
	\label{la16}\en	
	where $j$ is not necessarily positive and $q$ is not necessarily real. Actually, since (formally) $L_0^\dagger=-L_0$, its (standard) eigenvalues should be purely imaginary. Hence $q$ should be purely imaginary.
	
	\vspace{2mm}
	
	{\bf Remark:--} It is important to notice that the fact that $L_0$ is only formally self-adjoint can be understood in view of the results in \cite{bag2021JPA}, where it was proved that the eigenvalues of $H_{IHO}$ are indeed purely imaginary, but they corresponds to eigenvectors which are tempered distributions  rather than square-integrable functions. Hence, using $L_3=\frac{i}{4}H_{2}$, we see that, contrarily to what observed before, the eigenvalues of $L_3$ are indeed also real. This is infact what we get below, where we prove that these eigenvalues  are $-\left(l+\frac{1}{4}\right)$  and $-\left(l+\frac{3}{4}\right)$, $l=0,1,2,3,\ldots$. 
	
	\vspace{2mm}
	
	Some results on extended $\mathfrak{su}(1,1)$ can be found in the appendix \ref{appendixC}. Here we just observe that similar formulas as those in (\ref{la7}) are also deduced here:
		\be
	L_+\Psi_{j,q}=(q-j)\Psi_{j,q+1}, \qquad L_-\Psi_{j,q}=(q+j)\Psi_{j,q-1},
	\label{la17}\en
	at least if $L_\pm\Psi_{j,q}\neq0$.
	
	After this general introduction, let us see what these operators and vectors become with our specific choice in (\ref{la13}). In terms of the bosonic operators we find that
	\be
	L_0=L_3=-\frac{i}{4}\left({a^\dagger}^2+a^2\right), \qquad L_1=\frac{i}{4}\left(a^\dagger a+aa^\dagger\right), \qquad L_2=\frac{i}{4}\left({a^\dagger}^2-a^2\right),
	\label{la18}\en
	so that, in particular
	\be
	L_+=-\frac{1}{4}\left(a^\dagger-i\,a\right)^2=-\frac{1}{2}\,A^2, \qquad L_-=\frac{1}{4}\left(a^\dagger+i\,a\right)^2=-\frac{1}{2}\,B^2,
	\label{la19}\en
	where we have introduced the operators
	\be
	A=\frac{1}{\sqrt{2}}\left(a^\dagger-ia\right), \qquad B=\frac{i}{\sqrt{2}}\left(a^\dagger+ia\right).
	\label{la20}\en
	\be
	[A,B]=\1,
	\label{la21}\en
	(as usual, in the sense of unbounded operators). Since $B^\dagger\neq A$, these are pseudo-bosonic operators\footnote{More exactly, they are $\scr$-pseudo bosons, \cite{bagspringerbook}.}. In terms of these operators we can rewrite $L_0$ and $L^2$ as follows:
	\be
	L_0=-\frac{1}{2}\left(\hat n+\frac{1}{2}\1\right), \qquad L^2=-\frac{3}{16}\1,
	\label{la22}\en
	where $\hat n=BA$  is the pseudo-bosonic number operator. We see that $L^2=K^2$. Hence, in this particular case, $j\in\mathbb{R}$ and, to be more explicit, $j=-\frac{1}{4}$ or $j=-\frac{3}{4}$, as before. Repeating what we have done in the previous case, we can first find the eigenstates of $\hat n$ looking, first, to the vacuum of $A$: $A\varphi_0=0$, and then constructing a family of vectors $\varphi_n=\frac{B^n}{\sqrt{n!}}\varphi_0$, $n\geq0$.     	
    We know that (at least formally) $\hat n\varphi_n=n\varphi_n$, $n\geq0$. Now, equation $L_0\Psi_{j,q}=q\Psi_{j,q}$ in (\ref{la16}) produces, if we take $\Psi_{j,q}=\varphi_0$, 
	$$
	-\frac{1}{2}\left(\hat n+\frac{1}{2}\1\right)\varphi_0=-\frac{1}{4}\varphi_0,
	$$
	so that we can put $\Psi_{j_k,-\frac{1}{4}}=\varphi_0$, $k=1,2$. Similarly to what observed in the previous case, actng with $L_+$ and with its powers on $\Psi_{j_k,-\frac{1}{4}}$ destroys the state, since $L_+\Psi_{j_k,-\frac{1}{4}}\propto A^2\varphi_0=0$. On the other hand, acting on $\Psi_{j_k,-\frac{1}{4}}$ with $L_-$ and with its powers produce vectors which are proportional to $\varphi_{2l}$: $L_-^l\varphi_0\propto\varphi_{2l}$, $l\geq0$. Hence the eigenvalues of $L_0$ are the following:
	$$
	L_0\left(L_-^l\varphi_0\right)=-\left(l+\frac{1}{4}\right)\left(L_-^l\varphi_0\right),
	$$
	which is very similar to (\ref{la10}). Analogously, if we consider the vector $\varphi_1$, we observe that 
	$$
	-\frac{1}{2}\left(\hat n+\frac{1}{2}\1\right)\varphi_1=-\frac{3}{4}\varphi_1,
	$$
	so that we can put $\Psi_{j_k,-\frac{3}{4}}=\varphi_1$, $k=1,2$. Again, $L_+\Psi_{j_k,-\frac{3}{4}}\propto A^2\varphi_1=0$, while $L_-^l\varphi_1\propto\varphi_{2l+1}$, $l\geq0$, so that 
	$$
	L_0\left(L_-^l\varphi_1\right)=-\left(l+\frac{3}{4}\right)\left(L_-^l\varphi_1\right),
	$$
	which is very similar to (\ref{la12}).
%{\color{red}We are treating \(L_{+}\) as the lowering operator and \(L_{-}\) as the raising operator. But if we notice that it's also true that \(A=B_{-}\) and \(B=-A_{-}\), we could work with \(L_{-}\) as lowering operator,  that i think it's more natural. Following the same procedure we will obtain the set of eigenstates \(\varphi_n^{(-)}\) instead of \(\varphi_n^{(+)}\). 
	
%The set \(\varphi_n^{(+)}\) can be obtained by treating \(L_{+}\) as the lowering operator, as done above, or by changing appropriately the signs of the generators.
%For example, by changing properly the sign of the generators should be possible to invert \(L_{\pm}\), in such a way that the new \(\tilde{L}_{\pm}=L_{\mp}\). We would have \(\tilde{L}_{-}=L_{+}=-\frac{1}{2}A^2=-\frac{1}{2}(-iA_{+})^2\), and by using the same procedure we would obtain \(\varphi_n^{(+)}\) }

What we have seen so far is apparently not very different from what we have done with the choce (\ref{la2}) of the generators. However, there is indeed a big difference, which clearly emerges when we write down the explicit expression for $\varphi_0$ in the $x$ coordinates. Indeed, if we look for the solution of the equation $A\varphi_0=0$, it is straightforward to rewrite this as $\varphi_0'(x)=\frac{1-i}{1+i}x\varphi_0(x)$, whose solution is simply $\varphi_0(x)=Ne^{-i\frac{x^2}{2}}$, which is not in $\ltwo$. This is in full agreement with what we have shown in (\ref{54}). Indeed the $\varphi_0(x)$ we have just found is exactly the same as in (\ref{54}), for $n=0$, $\Omega=1$ and $N=\frac{1}{\pi^{1/4}}$. The other vectors introduced here, the eigenstates of $L_0 $ $\varphi_n$, are exactly those in (\ref{54}).

\vspace{2mm}

It is clear that another possible choice of the generators is also possible, the one in which, rather than on the eigenstates of $H_0$ or $H_2$, we focus on those of $H_1$. But this analysis does not add much to what we have discussed so far and will not be given.

\vspace{2mm}

{\bf Remark:--} It is possible to use the results given in (\ref{la1}) to deduce a (formal) expression of the operators $T_{\pm}$ in (\ref{61}). Indeed, from the latter and from the eigenvalue equations $H_0\varphi_n=E_n\varphi_n$ and $H_2\varphi_n^{(\pm)}=\pm i E_n\varphi_n^{(\pm)}$ we can easily deduce the following relation: $T^{-1}_\pm H_2T_\pm=\pm i H_0$.

Then, since
$$
e^{-\gamma H_1}H_2e^{\gamma H_1}=H_2\cos(4\gamma)-iH_0\sin(4\gamma),
$$
we conclude that, fixing  \(\gamma=\pm\frac{\pi}{8}\), we have
$$
e^{\mp\frac{\pi}{8} H_1}H_2e^{\pm\frac{\pi}{8} H_1}=\mp iH_0,
$$
and therefore we conclude that, at least formally, \(T_{\pm}=e^{\mp\frac{\pi}{8} H_1}\).

%
%
% The first simple remark is that the operators $T_\pm$ cannot be of the form $e^{i\gamma H_1}$, because of the lack of possible choice of $\gamma$. Indeed, we can easily find that
%$$
%e^{-i\gamma H_1}H_2e^{i\gamma H_1}=H_2\cosh(4\gamma)+H_0\sinh(4\gamma)
%$$
%as well as $e^{-i\gamma H_1}H_0e^{i\gamma H_1}=H_0\cosh(4\gamma)+H_2\sinh(4\gamma)$. Our claim is based on the fact that $\cosh(4\gamma)\neq0$ for all possible choices of $\gamma$. We also find that.....{\color{red}Emanuele: could you please check the form of this operator? I guess is a bit long, but not complicated}
%{\color{blue}
%Non ho ben capito se mi chiede di verificare la correttezza dell'equazione del Remark o se mi chiede di trovare la forma degli operatori \(T_{\pm}\), per sicurezza ho fatto entrambe le cose. L'equazione precedente dai miei conti sembra corretta ma le allego i conti in un altro foglio. Per quanto riguarda la forma degli operatori \(T_{\pm}\) formalmente il ragionamento è il seguente. (Allego sempre conti nel foglio)
%$$
%e^{-\gamma H_1}H_2e^{\gamma H_1}=H_2\cos(4\gamma)-iH_0\sin(4\gamma)
%$$
%when \(\gamma=\frac{\pi}{8}\) we have
%$$
%e^{\mp\frac{\pi}{8} H_1}H_2e^{\pm\frac{\pi}{8} H_1}=\mp iH_0
%$$
%that suggest, at least formally, \(T_{\pm}=e^{\mp\frac{\pi}{8} H_1}\)
%}

\section{Conclusions}\label{sect6}

In this paper, we have proposed a unifying view to three Hamiltonians which are particularly interesting in quantum mechanics, $H_{BK}$, $H_{HO}$ and $H_{IHO}$. In particular, we have considered their eigenvectors and eigenvalues, and we have shown that $\ltwo$ is not enough in this analysis. This conclusion is confirmed both using a functional analysis point of view to the system and by considering a different Lie-algebraic approach. In this particular case, we have shown how to use $\mathfrak{su}(1,1)$ 
to construct a unifying view for the three Hamiltonians, but we have also shown that, depending on what we choose to diagonalize (and this is reflected by our definitions in (\ref{la2}) and (\ref{la13})),  $\ltwo$ can be the relevant space or not. 

This analysis is, we believe, interesting itself, but it is also a preliminary step towards a better understanding of the role of $H_{BK}$ in the study of the zeros of the Riemann function.  This is part of our future work, together with more mathematical results and connections with physics.

	\section*{Acknowledgments}  F. B. acknowledges partial support from Palermo University and from G.N.F.M. of the INdAM. F, B. also  acknowledges partial support from PRIN Project {\em  Transport phonema in low dimensional structures: models, simulations and theoretical aspects}, and partial support  by  the project ICON-Q, Partenariato Esteso NQSTI - PE00000023, Spoke 2.

 \appendix
	
	\renewcommand{\theequation}{A.\arabic{equation}}

		\renewcommand{\theequation}{A.\arabic{equation}}
	
	\section{Appendix: an interesting space}\label{appendixB}
	
	 To show this, we start introducing the following set 
	\begin{equation}
		\mathcal{G}=\{f:\mathbb{R}\rightarrow\mathbb{C} \quad s.t. \quad f(x)=F(z)|_{\mathbb{R}}, \quad F:\mathbb{C}\rightarrow\mathbb{C} \quad \mbox{analytic}\}.
	\end{equation}
	This space is a vector space, and it is also an algebra, since the product of analytic functions is analytic. However, it is not a *-algebra since the complex conjugate of an analytic function is no longer analytic, in general.
	
	We can endow this space with a topology \(\tau_{\mathcal{G}}\) given by the seminorms
	\(p_{n}(f)=sup_{z\in \overline{B(0,n)}}|F(z)|\), where \(\overline{B(0,n)}=\{z\in \mathbb{C} \quad s.t.\quad |z| \leq n\}\). Indeed, it is easy to check that each $p_n(.)$ is a seminorm. It is also possible to show that the family \(\{p_n\}_{n\in \mathbb{N}}\) is separating, i.e that $p_n(f)=0$  $\forall n\in \mathbb{N}$ if and only if $f=0$. For that we first observe that, if \(f(x)=0\), then the uniqueness of the analytic extension implies that \(F(z)=0\), so that $p_n(f)=0$ for all $n$. Viceversa, if  $p_n(f)=0$ $\forall n$, then $F(z)=0$ for all $z\in\mathbb{C}$. Hence $f(x)=0$.

	Next we have the following result:
	\begin{lemma}
		\((\mathcal{G},\tau_\mathcal{F})\) it's a Frechét space
	\end{lemma}
	\begin{proof}
		We have only to prove that the space is complete. Let's consider a Cauchy sequence, \(f_n \in \mathcal{F}\), so \(p_n(f_k-f_l)=sup_{z\in \overline{B(0,n)}}|F_k(z)-F_l(z)|\rightarrow 0  \quad \forall n \in \mathbb{N}\). From this we can conclude that \(F_k\) is a sequence of analytic functions that converges uniformly on every compact subsets  of \(\mathbb{C}\) to a function \(F(z)\). From the Weierstrass theorem, it follows that \(F(z)\) is analytic on \(\mathbb{C}\). 
        The result follow by taking \(f(x)=F(z)|_{\mathbb{R}}\). 
	\end{proof}
	A key feature of this space is that it includes a subspace which is dense in \(L^2(\mathbb{R})\). Indeed we have
	\begin{prop}
		\(\mathcal{G}\) contains a dense subspace of \(L^2(\mathbb{R})\).
	\end{prop}
	\begin{proof}
		Consider \(F(z) = e^{-z^2/2}\) and let \(H_n(z)\) denote the complex Hermite polynomials. They are both analytic functions, so that \(e_n(z) = F(z) H_n(z)\) is analytic as well, for all $n$. Hence, \(e_n(x) = e_n(z)\vert_{\mathbb{R}} \in \mathcal{G}\). But these functions are, a part from a normalization factor, the eigenstates of the harmonic oscillator which form an o.n. basis in $\ltwo$. Then the linear span \(\mathrm{l.s.}\{ e_n(x) \}\subset \mathcal{G}\) is dense in \(L^2(\mathbb{R})\).
	\end{proof}
	The set $\G$ is stable under the action of the multiplication and of the derivation operator. Therefore, $\G$ is stable under the action of $H_j$, $j=0,1,2$.

			\renewcommand{\theequation}{B.\arabic{equation}}
	
	\section{Appendix: extended $\mathfrak{su}(1,1)$}\label{appendixC}
	
	Let us consider three operators, $x_\pm$ and $x_0$ satisfying (in the sense of unbounded operators, if needed)  $[x_0,x_\pm]=\pm x_\pm$, and $[x_+,x_-]=-2x_0$, but with $x_+^\dagger\neq x_-$ and $x_0^\dagger\neq x_0$. 
	
	First we put, using the same abuse of notation as in $\mathfrak{su}(1,1)$, 
	\be
	x^2=x_0^2-\frac{1}{2}(x_+x_-+x_-x_+)=x_0^2+x_0-x_-x_+=x_0^2-x_0-x_+x_-.
	\label{c1}\en
The operator $x^2$ commutes with each $x_\alpha$: $[x^2,x_\alpha]=0$, for $\alpha=0,\pm$. In particular $x^2$ and $x_0$ commute, we can look for common eigenstates of these two operators. Using again the same notation adopted for ordinary $\mathfrak{su}(1,1)$, we assume the following: there exists a non zero vector $\Psi_{j,q_0}\in\D$ (the domain of all the operators relevant for us) satisfying the following eigenvalue equations:
	\be
	\left\{
	\begin{array}{ll}
		x^2\Psi_{j,q_0}=j(j+1)\Psi_{j,q_0},\\
		x_0\Psi_{j,q_0}=q_0\Psi_{j,q_0},\\
	\end{array}
	\right.
	\label{c2}\en
	for some $j$ and $q_0$. Of course, there is no reason a priori to assume here that $j$ and $q_0$ are real or positive. This is because, as already observed, $x^2$ is not positive or self-adjoint, and $x_0$ is not self-adjoint, too. This is the main difference with respect to ordinary $\mathfrak{su}(1,1)$. However, some useful result can still be found. Indeed, using the commutation rules for $(x^2,x_\alpha)$ we deduce that
	\be
	\left\{
	\begin{array}{ll}
		x^2(x_\pm\Psi_{j,q_0})=j(j+1)(x_\pm\Psi_{j,q_0}),\\
		x_0(x_\pm\Psi_{j,q_0})=(q_0\pm1)(x_\pm\Psi_{j,q_0}),\\
	\end{array}
	\right.
	\label{c3}\en
	at least if $\Psi_{j,q_0}\notin \ker(x_\pm)$. This means that $x_\pm$ are ladder operators and, in particular, that $x_+$ is a raising while $x_-$ is a lowering operator. Using the same standard arguments for $\mathfrak{su}(1,1)$, we can also deduce that
	\be
	\left\{
	\begin{array}{ll}
		x_+\Psi_{j,q_0}=(q_0-j)\Psi_{j,q_0+1},\\
		x_-\Psi_{j,q_0}=(q_0+j)\Psi_{j,q_0-1}.\\
	\end{array}
	\right.
	\label{c4}\en
	These equations are in agreement with the fact that, as it is easy to check,
	$$
	[x_0,x_-x_+]=[x_0,x_+x_-]=0.
	$$
	In fact, from (\ref{c3}) we see that $x_0$ and $x_-x_+$ have the same eigenvectors. The same is true for  $x_0$ and $x_+x_-$.

	We have several possibilities: 
	
	\vspace{2mm}
	
	{\bf Case 1:--} for some $m\in\mathbb{N}_0=\mathbb{N}\cup\{0\}$ we have $x_-^{m-1}\Psi_{j,q_0}\neq0$ and $x_-^{m}\Psi_{j,q_0}=0$. In this case the set of eigenvalues of $x_0$, $\sigma(x_0)$, is bounded below: $\sigma(x_0)=\{q_0-m+1,q_0-m+2,q_0-m+3,\ldots\}$.

	\vspace{2mm}
	
	{\bf Case 2:--} for some $k\in\mathbb{N}_0=\mathbb{N}\cup\{0\}$ we have $x_+^{k-1}\Psi_{j,q_0}\neq0$ and $x_+^{k}\Psi_{j,q_0}=0$. In this case $\sigma(x_0)$, is bounded above: $\sigma(x_0)=\{\ldots,q_0+k-3,q_0+k-2,q_0+k-1\}$.

	\vspace{2mm}
	
	{\bf Case 3:--} both conditions above are true: for some $m\in\mathbb{N}_0=\mathbb{N}\cup\{0\}$ we have $x_-^{m-1}\Psi_{j,q_0}\neq0$ and $x_-^{m}\Psi_{j,q_0}=0$, and  for some $k\in\mathbb{N}_0=\mathbb{N}\cup\{0\}$ we have $x_+^{k-1}\Psi_{j,q_0}\neq0$ and $x_+^{k}\Psi_{j,q_0}=0$. In this case, of course, $\sigma(x_0)$, is bounded above and below: $\sigma(x_0)=\{q_0-m+1,q_0-m+2,\ldots,q_0+k-2,q_0+k-1\}$.
	
	\vspace{2mm}
	
	{\bf Case 4:--} neither Case 1, nor Case 2, hold. Then $\sigma(x_0)$ has no bound below and above.

	\end{document}